\documentclass{patmorin}
\usepackage{pat}
\usepackage[T1]{fontenc}
\usepackage[utf8]{inputenc}
\usepackage{todonotes,float,subcaption}
\usepackage{algorithm,algpseudocode}

\usepackage[mathlines]{lineno}
\newcommand*\patchAmsMathEnvironmentForLineno[1]{%
 \expandafter\let\csname old#1\expandafter\endcsname\csname #1\endcsname
 \expandafter\let\csname oldend#1\expandafter\endcsname\csname end#1\endcsname
 \renewenvironment{#1}%
    {\linenomath\csname old#1\endcsname}%
    {\csname oldend#1\endcsname\endlinenomath}}%
\newcommand*\patchBothAmsMathEnvironmentsForLineno[1]{%
 \patchAmsMathEnvironmentForLineno{#1}%
 \patchAmsMathEnvironmentForLineno{#1*}}%
\AtBeginDocument{%
\patchBothAmsMathEnvironmentsForLineno{equation}%
\patchBothAmsMathEnvironmentsForLineno{align}%
\patchBothAmsMathEnvironmentsForLineno{flalign}%
\patchBothAmsMathEnvironmentsForLineno{alignat}%
\patchBothAmsMathEnvironmentsForLineno{gather}%
\patchBothAmsMathEnvironmentsForLineno{multline}%
}

\usepackage[longnamesfirst,numbers,sort&compress]{natbib}

\definecolor{brightmaroon}{rgb}{0.76, 0.13, 0.28}
\definecolor{linkblue}{rgb}{0, 0.337, 0.227}
\newcommand{\defin}[1]{\emph{\textcolor{brightmaroon}{#1}}}
\makeatletter
\def\mathcolor#1#{\@mathcolor{#1}}
\def\@mathcolor#1#2#3{%
  \protect\leavevmode
  \begingroup
    \color#1{#2}#3%
  \endgroup
}
\makeatother

\usepackage[inline]{enumitem}
\newcommand{\sm}{\smallsetminus}
\newenvironment{clmproof}{\begin{proof}[Proof of Claim:]}{\end{proof}}

\title{\MakeUppercase{$3$-packings in Triangulations: Algorithms, bounds, and Complexity}}

\author{
Prosenjit Bose\thanks{School of Computer Science, Carleton University, Ottawa, Canada. Emails: {\tt jit@scs.carleton.ca},  {\tt anil@scs.carleton.ca}, and {\tt bobby.miraftab@gmail.com}}
\;\;\;\;\;\;Anil Maheshwari\footnotemark[1]\;\;\;\;\;\;Bobby Miraftab\footnotemark[1]\;\;\;\;\;\;
Yota Otachi\thanks{Nagoya University, Nagoya, Japan. Email: {\tt otachi@nagoya-u.jp}.}
}

\begin{document}
\maketitle

\begin{abstract}
We study $H$-packings in plane triangulations for the three-vertex graphs
$H\in\{P_3,K_3,P_2\cup P_1\}$. For a graph $H$, let $\lambda_H(G)$ denote
the maximum size of an $H$-packing in $G$, with the convention that for
$H=P_2\cup P_1$ the copies are required to be induced. For $P_3$-packings,
we prove that every triangulation $G$ on $n$ vertices satisfies
$\lambda_{P_3}(G)\ge \left\lfloor \frac n5\right\rfloor$,
and show that this lower bound is asymptotically tight. 
We also study
triangle packings in triangulations and provide lower bounds for
$\lambda_{K_3}(G)$ in terms of the maximum degree and the degree sequence.
We give a face-path characterization of triangle factors in $4$-connected
plane triangulations using a hamiltonian cycle and the weak duals of the two
associated maximal outerplanar graphs. Finally, for induced packings by
$P_2\cup P_1$, we prove that every plane triangulation $T$ on $n$ vertices
satisfies $\lambda_{P_2\cup P_1}(T)\ge
\left\lfloor \frac n3\right\rfloor-2$,
and show that such a packing can be found in polynomial time.
\end{abstract}

Packing or tiling problems ask for large collections of vertex-disjoint copies of a fixed pattern inside a host graph. 
More precisely for graphs $H$ and $G$, an \defin{$H$-packing} ( also called an \defin{$H$-tiling}), in $G$ is a family of pairwise vertex-disjoint subgraphs of $G$, each isomorphic to $H$. 
The objective in the optimization version is to maximize the size of such a family. The spanning analogue is a \defin{perfect $H$-packing}, also known as an \defin{$H$-factor}.
The case $H=K_2$ recovers ordinary matching, placing $H$-packing at the intersection of matching theory, extremal combinatorics, and algorithmic graph theory. From the computational perspective, this generalization becomes difficult as soon as one moves beyond $K_2$. In the language of generalized matchings and $G$-partitions, Hell and Kirkpatrick \cite{HellKirkpatrick1981,KirkpatrickHell1983} showed that if the pattern has a component on at least three vertices, then deciding whether the vertex set of a graph can be partitioned into copies of the pattern is NP-complete.
When all components have order at most two, the problem is solvable by matching-type methods. In contrast, dense-graph theory provides broad structural conditions guaranteeing perfect packings.
In fact the seminal work on $H$-factors in dense graphs by \citet{alon1996h} and subsequent sharp threshold results by \citet{KuhnOsthus2009} show that sufficiently large minimum degree forces a perfect $H$-tiling, with the asymptotically correct threshold determined by chromatic-type parameters of $H$.
Among the smallest non-trivial patterns are the three-vertex graphs $P_3$, $K_3$, and $P_2\cup P_1$. Despite their simplicity, the corresponding packing and factor problems already exhibit diverse behaviour, ranging from hardness phenomena in sparse graph classes to structured regimes where planar constraints can be exploited.

The $P_3$ case has inspired a substantial line of work on \defin{3-vertex path packings}, often called $\Lambda$-packings, and spanning $\Lambda$-factors. \citet{p3} developed structural and extremal results for packing $3$-vertex paths, and Kelmans \cite{Kelmans2011} further advanced the theory in claw-free graphs, where additional local structure enables stronger packing and factor statements. These results show how restrictions on the host graph can substantially change the behaviour of small-pattern packings.

For triangles, this line of research traces back to the theorem of Corradi and Hajnal \cite{CorradiHajnal1963}, which initiated the study of minimum-degree conditions forcing large collections of disjoint cycles and triangles. \citet{CapraraRizzi2002} refined the computational picture by showing that maximum node-disjoint triangle packing is polynomial-time solvable when $\Delta(G)\le 3$, while NP-hardness already holds for planar graphs of maximum degree $4$. More recently, in parameterized complexity, \citet{ShengXiao2022} obtained an improved linear kernel for planar vertex-disjoint triangle packing, reducing the kernel size from a previously known $732k$ bound to $141k$ vertices.

In this paper, we study these three-vertex packing problems in the particularly rigid planar class of triangulations. Triangulations are sparse in the planar sense, but locally dense, as every face is a triangle, and their facial structure gives a natural source of candidate copies of $K_3$. 
They also have a strong global structure. 
For example, $4$-connected triangulations are hamiltonian by Tutte's theorem, and stronger results show that in a $4$-connected plane graph one can often force hamiltonian cycles through prescribed edges; see \cite{Whitney1932,Tutte1956,Sanders1996}. 
These features make triangulations a natural setting in which to seek sharp packing theorems and useful algorithmic descriptions.

Our first main result concerns $P_3$-packings. 
We prove that every triangulation $G$ on $n$ vertices contains a $P_3$-packing of size at least $\left\lfloor n/5\right\rfloor$.
The proof uses Barnette's theorem, which guarantees a spanning tree of maximum degree at most $3$ in every $3$-connected planar graph, together with a simple packing lemma for subcubic trees. 
We also construct an infinite family of triangulations showing that the constant $1/5$ is asymptotically best possible.

We then study triangle packings. For a plane triangulation $G$, we introduce an auxiliary graph whose vertices are the facial triangles of $G$, with two vertices adjacent whenever the corresponding facial triangles share a vertex in $G$. 
Independent sets in this auxiliary graph correspond exactly to families of pairwise vertex-disjoint facial triangles. 
Applying the Caro--Wei bound to this auxiliary graph gives degree-sensitive lower bounds on the maximum size of a triangle packing. 
In particular, if $\Delta=\Delta(G)$, then we get 
\[
\nu_3(G)\ge
\left\lceil
\frac{2n-4}{3\Delta-5}
\right\rceil ,
\]
and we also derive refined versions depending on the degree sequence of $G$. These bounds are constructive and lead directly to polynomial-time algorithms for finding such packings.

Next we consider triangle factors. 
In a $4$-connected plane triangulation $T$, every triangle is facial. Using the existence of a hamiltonian cycle of $T$, we decompose the triangulation into two maximal outerplanar graphs. The weak duals of these two graphs are trees, and the facial triangles incident with a fixed vertex form a path in the appropriate weak dual. This gives a characterization of triangle factors in terms of selecting vertices from two families of face-paths so that every original vertex is covered exactly once.

Finally, we study induced packings by $P_2\cup P_1$. 
Here we benefit from the structure of the complement graph.
In fact a triple induces $P_2\cup P_1$ in a triangulation $T$ precisely when it induces $P_3$ in $\overline T$. Using a matching theorem for planar graphs and a planar bipartite obstruction argument, we prove that every plane triangulation $T$ on $n$ vertices satisfies
\[
\lambda_{P_2\cup P_1}(T)
\ge
\left\lfloor \frac n3\right\rfloor -2 .
\]
We also show that such a packing can be found in polynomial time.

\subsection{Overview}

The paper is organized as follows. After the preliminaries, we prove the lower bound for $P_3$-packings in triangulations and discuss $P_3$-factors. We then study $K_3$-packings through the auxiliary graph of facial triangles, deriving degree-based lower bounds and algorithmic consequences. Next we turn to triangle factors in $4$-connected triangulations. We conclude with induced $(P_2\cup P_1)$-packings.


\section{Preliminaries}

All graphs in this paper are finite, connected. For a graph $G$, we write
$V(G)$ and $E(G)$ for its vertex set and edge set, and we call $|V(G)|$ the
\defin{order} of $G$. For a vertex $v\in V(G)$, let $d_G(v)$ denote its
degree and $N_G(v)$ its neighbourhood. We write $\Delta(G)$ and $\delta(G)$
for the maximum and minimum degree of $G$, respectively. When the graph is
clear from context, we write $d(v)$, $N(v)$, $\Delta$, and $\delta$.
For $S\subseteq V(G)$, let $G[S]$ denote the subgraph of $G$ induced by $S$.
The complement of $G$ is denoted by $\overline G$. A set $S\subseteq V(G)$ is
\defin{stable}, or \defin{independent}, if no two vertices of $S$ are adjacent.
We write $\alpha(G)$ for the maximum size of a independent set  in $G$.
For graphs $H$ and $G$, an \defin{$H$-packing} in $G$ is a family of pairwise
vertex-disjoint subgraphs of $G$, each isomorphic to $H$. An $H$-packing is
\defin{perfect}, or an \defin{$H$-factor}, if it covers every vertex of $G$.
A copy of a graph $F$ in $G$ is \defin{induced} if it is an induced subgraph of
$G$ isomorphic to $F$.
A \defin{plane graph} is a planar graph together with a fixed embedding. The
set of faces of a plane graph $G$ is denoted by $F(G)$. A \defin{facial
triangle} is a facial cycle of length three.
A \defin{plane triangulation} is a plane graph in which every face, including
the outer face, is bounded by a triangle. 
In a plane triangulation, every non-facial triangle is separating; in
particular, every triangle of a $4$-connected plane triangulation is facial.
A cycle $C$ in a plane graph is \defin{separating} if both the interior and
the exterior of $C$ contain vertices of the graph. A \defin{separating
triangle} is a separating cycle of length three.
A graph is \defin{$k$-connected} if it has more than $k$ vertices and remains
connected after deleting any set of fewer than $k$ vertices. 
A \defin{hamiltonian cycle} of a graph $G$ is a cycle containing every vertex of $G$.
A plane graph is \defin{outerplanar} if it has an embedding in which all
vertices lie on the outer face. It is \defin{maximal outerplanar} if no edge
can be added while preserving outerplanarity; equivalently, every bounded
face is a triangle.
For a plane graph $G$, the \defin{dual graph} $G^*$ has one vertex for each face of $G$, with two vertices adjacent whenever the corresponding faces of $G$ share an edge. If $G$ has a distinguished outer face, the \defin{weak dual} of $G$ is the subgraph of $G^*$ induced by the vertices corresponding to the bounded faces of $G$.

\section{\texorpdfstring{$P_3$-packings}{P3-packings}}

We begin with packings by $3$-vertex paths. For a graph $G$, a
\defin{$P_3$-packing} in $G$ is a collection of pairwise vertex-disjoint
subgraphs of $G$, each isomorphic to $P_3$. We write $\lambda(G)$ for the
maximum size of a $P_3$-packing in $G$.

We start with the Barnette's theorem.
The role of Barnette's theorem in our argument is to reduce the problem from triangulations to trees of bounded degree. 
Indeed, if a graph $G$ contains a spanning subgraph $T$, then every $P_3$-packing in $T$ is also a $P_3$-packing in $G$. 
Thus, for a triangulation $G$, it is enough to find a
large $P_3$-packing in a suitable spanning tree of $G$. Barnette's theorem guarantees such a spanning tree with maximum degree at most $3$. We are therefore led to the following elementary packing lemma for subcubic trees.

\begin{lem}{\rm\cite[]{bar}}
Every $3$-connected planar graph has a spanning tree of maximum degree at most $3$.
\end{lem}

\begin{lem}\label{lem:packing}
If $T$ is a tree on $m$ vertices with $\Delta(T)\le 3$, then $T$ contains a $P_3$-packing of size at least
$\left\lfloor m/5\right\rfloor$.
\end{lem}

\begin{proof}
We argue by induction on $m$. The claim is trivial for $m<5$.
Root $T$ at a leaf. For a vertex $v$, let $T_v$ denote the rooted
subtree induced by $v$ and all its descendants. Choose $v$ as far from the root
as possible subject to $|T_v|\ge 3$. Then every child-subtree of $v$ has order
at most $2$. Since the root is a leaf and $m\ge 5$, the vertex $v$ is not the
root. Hence $v$ has at most two children, because $\Delta(T)\le 3$. 
Therefore $3\le |T_v|\le 1+2+2=5$, see \Cref{fig:tv}.
The subtree $T_v$ is connected and has at least three vertices, so it contains
a copy of $P_3$.

Let $s=|T_v|$. Delete $T_v$ from $T$. The remaining graph is either empty or a
tree $T'$ with maximum degree at most $3$, and $|T'|=m-s$. By induction, $T'$
contains a $P_3$-packing of size at least $\left\lfloor \frac{m-s}{5}\right\rfloor$.
Adding one copy of $P_3$ from $T_v$, we get a packing of size at least $1+\left\lfloor \frac{m-s}{5}\right\rfloor $.
Since $s\le 5$,
\[
1+\left\lfloor \frac{m-s}{5}\right\rfloor
\ge
\left\lfloor \frac m5\right\rfloor .
\]
This proves the lemma.
\end{proof}

\begin{figure}[H]
    \centering
    \includegraphics[scale=0.8]{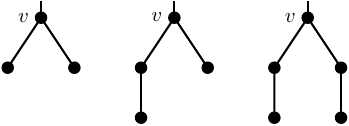}
    \caption{Possible cases for $T_v$}
    \label{fig:tv}
\end{figure}

We now combine the two previous lemmas to obtain the desired lower bound for triangulations.

\begin{thm}\label{thm:p3-packing-lower}
Let $G$ be a triangulation of order $n$. Then $\lambda(G)\ge \left\lfloor \frac n5 \right\rfloor$.
\end{thm}

\begin{proof}
If $n=3$, then the claim is trivial, since
$\left\lfloor 3/5\right\rfloor=0$. 
Hence assume $n\ge 4$.
Since every triangulation of order at least $4$ is $3$-connected, Barnette's
theorem gives a spanning tree $T$ of $G$ with $\Delta(T)\le 3$.
By \Cref{lem:packing}, the tree $T$ contains a $P_3$-packing of size at least
\[
\left\lfloor \frac{|V(T)|}{5}\right\rfloor
=
\left\lfloor \frac n5\right\rfloor .
\]
Since $T$ is a subgraph of $G$, this is also a $P_3$-packing in $G$. 
Therefore $\lambda(G)\ge \left\lfloor \frac n5\right\rfloor$.
\end{proof}

We first record a simple matching fact that will be used in the construction.
The point is that we will later need to associate to each original vertex
$v\in V(H)$ a distinct face incident with $v$. This allows us to place new
vertices in separate faces, so that the copies of $P_3$ built from different
original vertices are vertex-disjoint.

\begin{lem}{\rm \cite[Lemma 24]{GoddardHenning2016}}\label{lem:hall}
Let $H$ be a simple plane triangulation with $|V(H)|=k\ge 4$.
Then one can choose, for every vertex $v\in V(H)$, a face $f_v$ incident
with $v$ such that the faces $f_v$ are pairwise distinct.
\end{lem}

\begin{lem}\label{lem:p3-tight-family}
For every integer $k\ge 4$, there exists a triangulation $G_k$ of order
$n=5k-8$ such that
\[
\lambda(G_k)=k=\frac{n+8}{5}.
\]
\end{lem}

\begin{proof}
Let $H$ be any triangulation on $k$ vertices, and let $X=V(H)$.
Since $H$ is a triangulation, it has $2k-4$ faces.
We first choose, for every vertex $v\in X$, a distinct face $f_v$ of $H$
incident with $v$, which is possible by \Cref{lem:hall}.
Now modify each face $f=abc$ of $H$ as follows. 
We add a new edge $x_fy_f$ inside $f$ and join it to $a,b,c$. 
If $f=f_v$ is the face assigned to
$v$, we choose the position of $y_f$ so that $y_f$ is adjacent to $v$.
Hence, for every $v\in X$, the vertices $x_{f_v},v,y_{f_v}$ contain a copy of $P_3$, see \Cref{fig:triangulation-construction} for an illustration.
Let the resulting graph be $G_k$. 
Since every face of $H$ was replaced by a
triangulated patch, $G_k$ is again a triangulation. Its number of vertices is
\[
|V(G_k)|
=
|V(H)|+2|F(H)|
=
k+2(2k-4)
=
5k-8.
\]
We now show that $\lambda(G_k)=k$.
First, the $k$ paths using the triples
$\{x_{f_v},v,y_{f_v}\}$, where $v\in X$, are pairwise vertex-disjoint,
because the faces $f_v$ are distinct. Therefore $\lambda(G_k)\ge k$.

Conversely, delete the original vertex set $X$. In each original face $f$ of
$H$, only the two vertices $x_f$ and $y_f$ remain, and they form a single edge.
There are no edges between vertices inserted in different faces. Hence
$G_k-X$ is a disjoint union of edges, and therefore contains no copy of $P_3$.
Thus every copy of $P_3$ in $G_k$ must contain at least one vertex of $X$.
Since the copies in a $P_3$-packing are vertex-disjoint, any such packing has
size at most $|X|=k$. Hence $\lambda(G_k)\le k$.
Combining the two inequalities gives $\lambda(G_k)=k$. Since
$|V(G_k)|=5k-8$, we have
\[
\lambda(G_k)=k=\frac{|V(G_k)|+8}{5}.
\]
This gives an infinite family of triangulations with $P_3$-packing number
asymptotic to $n/5$.
\end{proof}

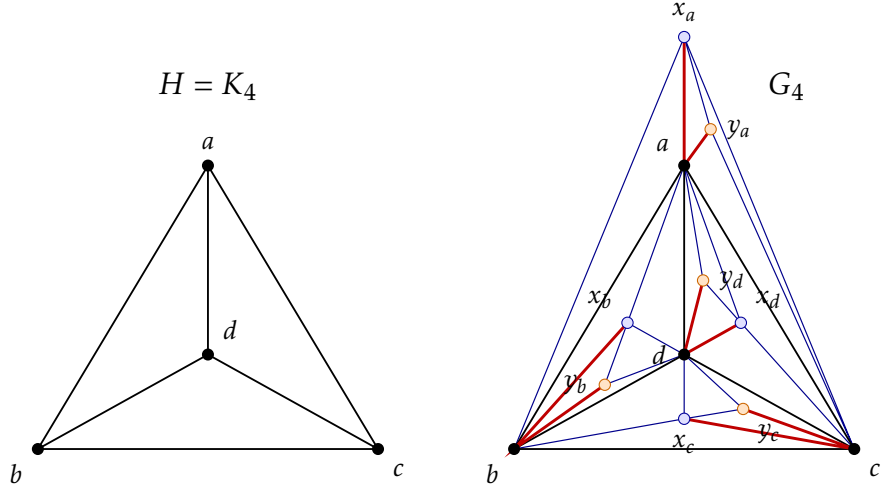
\begin{figure}[H]
\centering
\begin{tikzpicture}[
  x=1cm,y=1cm,
  every node/.style={font=\small},
  orig/.style={circle,fill=black,inner sep=1.6pt},
  xv/.style={circle,draw=blue!60!black,fill=blue!12,inner sep=1.5pt},
  yv/.style={circle,draw=orange!80!black,fill=orange!22,inner sep=1.5pt},
  oldedge/.style={black,line width=.7pt},
  newedge/.style={blue!55!black,line width=.45pt},
  p3edge/.style={red!75!black,line width=1.1pt},
  facelabel/.style={font=\scriptsize,fill=white,inner sep=1pt}
]

\begin{scope}
  \node[font=\large] at (0,3.35) {$H=K_4$};

  \coordinate (a) at (0,2.3);
  \coordinate (b) at (-2.25,-1.45);
  \coordinate (c) at (2.25,-1.45);
  \coordinate (d) at (0,-0.20);

  \draw[oldedge] (a)--(b)--(c)--cycle
                 (a)--(d)--(b)
                 (d)--(c);

  \node[orig,label=above:$a$] at (a) {};
  \node[orig,label=below left:$b$] at (b) {};
  \node[orig,label=below right:$c$] at (c) {};
  \node[orig,label=above right:$d$] at (d) {};

\end{scope}

\begin{scope}[shift={(6.3,0)}]
  \node[font=\large,fill=white,inner sep=1pt] at (1.35,3.35) {$G_4$};

  \coordinate (a) at (0,2.3);
  \coordinate (b) at (-2.25,-1.45);
  \coordinate (c) at (2.25,-1.45);
  \coordinate (d) at (0,-0.20);

  \coordinate (xa) at (0,4.00);
  \coordinate (ya) at (0.35,2.78);

  \coordinate (xb) at (-0.75,0.22);
  \coordinate (yb) at (-1.05,-0.60);

  \coordinate (xc) at (0,-1.05);
  \coordinate (yc) at (0.78,-0.92);

  \coordinate (xd) at (0.75,0.22);
  \coordinate (yd) at (0.25,0.78);

  \draw[newedge] (xa)--(a) (xa)--(b) (xa)--(c)
                 (ya)--(xa) (ya)--(a) (ya)--(c);

  \draw[newedge] (xb)--(a) (xb)--(b) (xb)--(d)
                 (yb)--(xb) (yb)--(b) (yb)--(d);

  \draw[newedge] (xc)--(b) (xc)--(c) (xc)--(d)
                 (yc)--(xc) (yc)--(c) (yc)--(d);

  \draw[newedge] (xd)--(c) (xd)--(a) (xd)--(d)
                 (yd)--(xd) (yd)--(d) (yd)--(a);

  \draw[oldedge] (a)--(b)--(c)--cycle
                 (a)--(d)--(b)
                 (d)--(c);

  \draw[p3edge] (xa)--(a)--(ya);
  \draw[p3edge] (xb)--(b)--(yb);
  \draw[p3edge] (xc)--(c)--(yc);
  \draw[p3edge] (xd)--(d)--(yd);

  \node[orig,label=above left:$a$] at (a) {};
  \node[orig,label=below left:$b$] at (b) {};
  \node[orig,label=below right:$c$] at (c) {};
  \node[orig,label=left:$d$] at (d) {};

  \node[xv,label=above:$x_a$] at (xa) {};
  \node[yv,label=right:$y_a$] at (ya) {};

  \node[xv,label=above left:$x_b$] at (xb) {};
  \node[yv,label=left:$y_b$] at (yb) {};

  \node[xv,label=below:$x_c$] at (xc) {};
  \node[yv,label=below right:$y_c$] at (yc) {};

  \node[xv,label=above right:$x_d$] at (xd) {};
  \node[yv,label=right:$y_d$] at (yd) {};
\end{scope}

\end{tikzpicture}
\caption{A triangulation $H=K_4$ in the left and the triangulation $G_4$  in the right obtained by adding two vertices in each face. The red paths are the four disjoint copies of $P_3$, namely $x_{f_v}vy_{f_v}$ for $v\in V(H)$.}
\label{fig:triangulation-construction}
\end{figure}

\begin{op}\label{ques:complexity-p3-packing}
Determine the computational complexity of maximum $P_3$-packing in plane
triangulations.
More precisely, given a plane triangulation $G$ and an integer $k$, decide whether $\lambda(G)\ge k$.
Is this problem solvable in polynomial time, or is it NP-complete even for restricted classes of triangulations, such as 
Eulerian triangulations, or triangulations of bounded maximum degree?
\end{op}

\begin{thm}\label{cor:p3-packing-linear}
Let $G$ be a plane triangulation of order $n$. Then a $P_3$-packing of size at least $\left\lfloor \frac n5\right\rfloor$ can be found in linear time.
\end{thm}

\begin{proof}
If $n=3$, then the empty packing has the required size. Hence assume
$n\ge 4$. 
Since every triangulation of order at least $4$ is
$3$-connected, we  apply the canonical-ordering construction of
\cite[Theorem~1]{Biedl2013}. This gives, in linear time, a spanning tree
$T$ of $G$ with $\Delta(T)\le 3$.
It remains to find the packing in $T$ in linear time. Root $T$ at a leaf and process the vertices in postorder. 
For each vertex $v$, keep the connected residual subtree consisting of $v$ together with the residual subtrees returned by its children. 
Since $\Delta(T)\le 3$ and the root is a leaf, every
non-root vertex has at most two children, and each returned residual subtree has size at most $2$. Thus the residual subtree considered at any vertex has size at most $1+2+2=5$.
Whenever this residual subtree has at least three vertices, choose a copy of $P_3$ inside it, add this copy to the packing, and discard the whole residual subtree. If it has at most two vertices, return it to the parent.

The chosen copies of $P_3$ are pairwise vertex-disjoint. Moreover, each chosen
copy is charged to a discarded connected subtree of size at most $5$, and at the end fewer than three vertices remain. Therefore, if the algorithm outputs $q$ copies of $P_3$, then $n\le 5q+2$.
Hence
\[
q\ge \left\lceil \frac{n-2}{5}\right\rceil
\ge
\left\lfloor \frac n5\right\rfloor .
\]
Since $T$ is a subgraph of $G$, this is also a $P_3$-packing in $G$.
The construction of $T$ is linear, and the postorder packing procedure also examines each vertex and edge only a constant number of times. 
Thus the whole algorithm runs in linear time.
\end{proof}
\subsection{\texorpdfstring{$P_3$-factors}{P3-factors}}

A \defin{$P_3$-factor} of a graph $G$ is a $P_3$-packing that covers all
vertices of $G$. Equivalently, it is a spanning subgraph whose connected
components are all isomorphic to $P_3$.

By applying Hall's theorem to a bipartite graph obtained by replacing each
candidate centre by two copies, we get the following characterization.

\begin{thm}\label{thm:p3-factor-hall}
Let $G$ be a graph of order $n$ with $3\mid n$. Then $G$ has a $P_3$-factor if
and only if there exists a set $C\subseteq V(G)$ with $|C|=n/3$ such that, for
every $X\subseteq V(G)\sm C$, we have $2|N_G(X)\cap C|\ge |X|$.
\end{thm}

\begin{proof}
First suppose that $G$ has a $P_3$-factor $\mathcal P$. In each copy of
$P_3$, choose the middle vertex as its centre, and let $C$ be the set of all
centres. Set $D\coloneqq V(G)\sm C$. Then $|C|=n/3$ and 
$|D|=2|C|$.
Construct a bipartite graph $B$ with bipartition $(D,\widetilde C)$, where
$\widetilde C$ contains two copies $c^{(1)}$ and $c^{(2)}$ of each vertex
$c\in C$. Join $d\in D$ to both copies of $c$ whenever $dc\in E(G)$. Each
path $xcy$ in $\mathcal P$, with centre $c$, gives two matching edges from
$x$ and $y$ to the two copies of $c$. Hence $B$ has a perfect matching. By
Hall's theorem, for every $X\subseteq D$, $|X|\le |N_B(X)|$.
But $N_B(X)$ consists of the two copies of the vertices in $N_G(X)\cap C$.
Therefore $|N_B(X)|=2|N_G(X)\cap C|$,
and so $2|N_G(X)\cap C|\ge |X|$.

Conversely, suppose there exists $C\subseteq V(G)$ with $|C|=n/3$ such that $2|N_G(X)\cap C|\ge |X|$
for every $X\subseteq D\coloneqq V(G)\sm C$. 
Construct the same bipartite
graph $B$ with bipartition $(D,\widetilde C)$. Then, for every $X\subseteq D$,
we obtain $|N_B(X)|=2|N_G(X)\cap C|\ge |X|$.
Thus Hall's condition holds, so $B$ has a matching saturating $D$. Since
$|D|=2|C|=|\widetilde C|$, this matching is perfect.
Now identify the two copies of each vertex $c\in C$. If the copies of $c$ are
matched to $d_1,d_2\in D$, put the edges $cd_1$ and $cd_2$ into a spanning
subgraph $F$ of $G$. 
Then every vertex of $C$ has degree $2$ in $F$, and every
vertex of $D$ has degree $1$ in $F$.
Let $Q$ be a connected component of $F$. 
We write $k\coloneqq|V(Q)\cap C|$, and  $\ell\coloneqq|V(Q)\cap D|$.
Counting edges in $Q$ from the $C$-side and from the $D$-side gives $e(Q)=2k=\ell$.
Since $Q$ is connected, $e(Q)\ge |V(Q)|-1=k+\ell-1$.
Substituting $\ell=2k$ gives $2k\ge 3k-1$,
so $k\le 1$. Since $Q$ contains an edge, $k\ge 1$. Hence $k=1$, and therefore
$\ell=2$. Thus every component of $F$ has exactly one vertex of $C$ and two
vertices of $D$, and is therefore a copy of $P_3$. Hence $F$ is a
$P_3$-factor of $G$.
\end{proof}

\section{\texorpdfstring{$K_3$-packings}{k3-packing}}

We now study packings by copies of $K_3$. Since every face of a triangulation is a triangle, facial triangles provide a natural source of candidate members for a triangle packing. This suggests encoding the interaction among facial triangles in an auxiliary graph, where independent sets correspond to families of pairwise vertex-disjoint triangles. Using this viewpoint, we derive several lower bounds on the triangle packing number in terms of degree parameters of the triangulation.

\begin{defn}
A \defin{$K_3$-packing} in a graph $G$ is a family of pairwise vertex-disjoint subgraphs of $G$, each isomorphic to $K_3$. We write $\nu_3(G)$ for the maximum size of a $K_3$-packing in $G$.
\end{defn}

\begin{defn}
Let $G$ be a plane graph. Define the graph $H$ as follows:
\begin{itemize}
    \item $V(H)=F(G)$, the set of faces of $G$;
    \item two vertices $f_1,f_2\in V(H)$ are adjacent in $H$ if and only if the corresponding faces of $G$ share a vertex.
\end{itemize}
\end{defn}

\begin{exa}
Here, the plane graph is drawn with gray edges, and the auxiliary graph $H$ is drawn with black edges; see \Cref{fig:aux}.

\begin{figure}[H]
    \centering
    \includegraphics[scale=0.7]{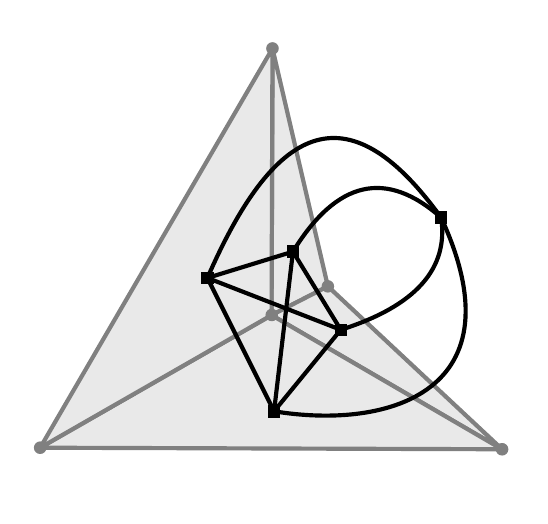}
    \caption{The plane graph $G$ in gray and the graph $H$ in black.}
    \label{fig:aux}
\end{figure}
\end{exa}

The advantage of this auxiliary graph is that it turns the packing problem into an independent set problem. 
Indeed, a set of facial triangles is pairwise vertex-disjoint in $G$ precisely when the corresponding vertices are independent in $H$. 
Thus any lower bound on $\alpha(H)$ immediately gives a lower bound on $\nu_3(G)$. 
In order to apply such a bound effectively, we first estimate the degrees in $H$ in terms of the degrees of the vertices of $G$.

\begin{lem}\label{lem:deg_H}
Let $G$ be a simple plane triangulation on $n>3$ vertices, and let
$\Delta=\Delta(G)$. If $f=xyz\in F(G)$, then $d_H(f)=d(x)+d(y)+d(z)-6$.
In particular, $\Delta(H)\le 3\Delta-6$.
\end{lem}

\begin{proof}
The vertex $x$ is incident with exactly $d(x)$ facial triangles, so there are
$d(x)-1$ facial triangles distinct from $f$ that contain $x$. Similarly, there
are $d(y)-1$ facial triangles distinct from $f$ that contain $y$, and
$d(z)-1$ facial triangles distinct from $f$ that contain $z$. Thus
\[
(d(x)-1)+(d(y)-1)+(d(z)-1)
\]
counts all neighbours of $f$ in $H$, except that the three facial triangles
adjacent to $f$ across the edges $xy$, $yz$, and $zx$ are each counted twice.
Therefore
\[
d_H(f)=(d(x)-1)+(d(y)-1)+(d(z)-1)-3
=d(x)+d(y)+d(z)-6.
\]
Since $d(x),d(y),d(z)\le \Delta$, we obtain $d_H(f)\le 3\Delta-6$
for every $f\in F(G)$, and hence $\Delta(H)\le 3\Delta-6$.
\end{proof}
We will now use this degree computation together with a standard lower bound on the independence number. The point is that the degrees in $H$ are not abstract.
In fact by~\Cref{lem:deg_H}, they are controlled by the degrees of the three vertices on the corresponding facial triangle of $G$.
Thus we use the well-known Caro--Wei bound, which states that every graph $G$ contains an independent set of size at least
\begin{equation}\label{caro}
\alpha(G) \ge \sum_{v \in V(G)} \frac{1}{d_G(v)+1}\ge \frac{|V(G)|}{\Delta(G)+1}.
\end{equation}
This result was obtained independently by Caro~\cite{Caro1979} and Wei~\cite{Wei1981}.

\begin{thm}\label{thm:degree-sensitive-k3-packing}
Let $G$ be a simple plane triangulation on $n>3$ vertices. Then
\[
\nu_3(G)\ge
\left\lceil
\sum_{xyz\in F(G)}
\frac{1}{d(x)+d(y)+d(z)-5}
\right\rceil .
\]
Consequently,
\[
\nu_3(G)\ge
\left\lceil
\frac{(2n-4)^2}
{\sum_{v\in V(G)}d(v)^2-5(2n-4)}
\right\rceil .
\]
\end{thm}

\begin{proof}
Let $H$ be the auxiliary graph whose vertices are the facial triangles
of $G$, where two vertices of $H$ are adjacent whenever the corresponding
facial triangles of $G$ share a vertex. Every independent set in $H$
therefore gives a family of pairwise vertex-disjoint facial triangles in
$G$. 
Hence $\nu_3(G)\ge \alpha(H)$.
Let $f=xyz$ be a facial triangle of $G$. 
By \Cref{lem:deg_H}, we have $d_H(f)=d(x)+d(y)+d(z)-6$,
and so we get $d_H(f)+1=d(x)+d(y)+d(z)-5$.
Applying the Caro--Wei bound to $H$, we obtain
\[
\alpha(H)
\ge
\sum_{f\in F(G)}\frac{1}{d_H(f)+1}
=
\sum_{xyz\in F(G)}
\frac{1}{d(x)+d(y)+d(z)-5}.
\]
Since $\nu_3(G)$ is an integer, the first bound follows.

Now we define $a(f)\coloneqq d(x)+d(y)+d(z)-5$ for each face $f=xyz$. 
By Cauchy--Schwarz, we obtain
\[
\sum_{f\in F(G)}\frac{1}{a(f)}
\ge
\frac{|F(G)|^2}{\sum_{f\in F(G)}a(f)}.
\]
Since $G$ is a plane triangulation on $n>3$ vertices, $|F(G)|=2n-4$.
Moreover, we have
\[
\sum_{f=xyz\in F(G)}\bigl(d(x)+d(y)+d(z)\bigr)
=
\sum_{v\in V(G)}d(v)^2,
\]
because each vertex $v$ is incident with exactly $d(v)$ facial triangles, and
each such incidence contributes $d(v)$. Therefore
\[
\sum_{f\in F(G)}a(f)
=
\sum_{v\in V(G)}d(v)^2-5(2n-4).
\]
This proves
\[
\nu_3(G)\ge
\left\lceil
\frac{(2n-4)^2}
{\sum_{v\in V(G)}d(v)^2-5(2n-4)}
\right\rceil .
\]
\end{proof}

The previous theorem can be rewritten in terms of the degree sequence of $G$. This gives a more explicit bound, and also shows how vertices of small degree improve the estimate.

\begin{cor}\label{cor:degree-sequence-k3-packing}
Let $G$ be a simple plane triangulation on $n>3$ vertices, let
$\Delta=\Delta(G)$, and let $n_i\coloneqq|\{v\in V(G):d(v)=i\}|$
for every integer $i\ge 3$. 
Then we have 
\[
\nu_3(G)\ge
\left\lceil
\frac{(2n-4)^2}
{(2n-4)(3\Delta-5)-\sum_{i=3}^{\Delta} i(\Delta-i)n_i}
\right\rceil .
\]
In particular,
\[
\nu_3(G)\ge
\left\lceil
\frac{(2n-4)^2}
{(2n-4)(3\Delta-5)-3(\Delta-3)n_3-4(\Delta-4)n_4}
\right\rceil .
\]
\end{cor}

\begin{proof}
By \Cref{thm:degree-sensitive-k3-packing},
\[
\nu_3(G)\ge
\left\lceil
\frac{(2n-4)^2}
{\sum_{v\in V(G)}d(v)^2-5(2n-4)}
\right\rceil .
\]
Since $G$ is a plane triangulation,
\[
\sum_{v\in V(G)}d(v)=6n-12=3(2n-4).
\]
Therefore
\[
\begin{aligned}
\sum_{v\in V(G)}d(v)^2-5(2n-4)
&=
\Delta\sum_{v\in V(G)}d(v)
-\sum_{v\in V(G)}d(v)(\Delta-d(v))
-5(2n-4)\\
&=
(2n-4)(3\Delta-5)
-\sum_{v\in V(G)}d(v)(\Delta-d(v))\\
&=
(2n-4)(3\Delta-5)
-\sum_{i=3}^{\Delta} i(\Delta-i)n_i.
\end{aligned}
\]
This proves the first bound.
For the second bound, observe that we have
\[
\sum_{i=3}^{\Delta} i(\Delta-i)n_i
\ge
3(\Delta-3)n_3+4(\Delta-4)n_4.
\]
Thus the denominator in the first bound is at most
$(2n-4)(3\Delta-5)-3(\Delta-3)n_3-4(\Delta-4)n_4$,
and the stated weaker bound follows.
\end{proof}

\begin{cor}\label{cor:old-delta-k3-packing}
Let $G$ be a simple plane triangulation on $n>3$ vertices, and let
$\Delta=\Delta(G)$. Then
\[
\nu_3(G)\ge
\left\lceil
\frac{2n-4}{3\Delta-5}
\right\rceil .
\]
\end{cor}

\begin{proof}
For every facial triangle $xyz$, we have  $d(x)+d(y)+d(z)-5\le 3\Delta-5$.
Therefore \Cref{thm:degree-sensitive-k3-packing} gives
\[
\nu_3(G)
\ge
\left\lceil
\sum_{xyz\in F(G)}\frac{1}{d(x)+d(y)+d(z)-5}
\right\rceil
\ge
\left\lceil
\frac{|F(G)|}{3\Delta-5}
\right\rceil .
\]
Since $|F(G)|=2n-4$, the result follows.
\end{proof}

\begin{cor}\label{cor:improved-delta-k3-packing}
Let $G$ be a simple plane triangulation on $n>3$ vertices, and let
$\Delta=\Delta(G)$. Then
\[
\nu_3(G)\ge
\left\lceil
\frac{(2n-4)^2}
{(3\Delta+8)n-12\Delta-16}
\right\rceil .
\]
\end{cor}

\begin{proof}
Since $G$ is a simple plane triangulation on $n>3$ vertices, every vertex
has degree at least $3$. Thus, for every $v\in V(G)$, we have $3\le d(v)\le \Delta$.
Hence we have $(d(v)-3)(d(v)-\Delta)\le 0$.
Expanding gives $d(v)^2\le (\Delta+3)d(v)-3\Delta$.
Summing over all vertices, we get 
\[
\sum_{v\in V(G)}d(v)^2
\le
(\Delta+3)\sum_{v\in V(G)}d(v)-3\Delta n.
\]
Since $G$ is a plane triangulation, we obtain 
\[
\sum_{v\in V(G)}d(v)=6n-12.
\]
Therefore
\[
\sum_{v\in V(G)}d(v)^2
\le
(\Delta+3)(6n-12)-3\Delta n.
\]
After subtracting $5(2n-4)$, we get
\[
\sum_{v\in V(G)}d(v)^2-5(2n-4)
\le
(3\Delta+8)n-12\Delta-16.
\]
Now apply \Cref{thm:degree-sensitive-k3-packing}.
\end{proof}

\begin{rmk}
The bound in \Cref{cor:improved-delta-k3-packing} is always at least as
strong as the bound in \Cref{cor:old-delta-k3-packing}. Indeed, the old
denominator, after putting the old bound over numerator $(2n-4)^2$, is $(2n-4)(3\Delta-5)$,
whereas the improved denominator is $(3\Delta+8)n-12\Delta-16$.
Their difference is
\[
(2n-4)(3\Delta-5)-\bigl((3\Delta+8)n-12\Delta-16\bigr)
=
3n(\Delta-6)+36.
\]
Since
\[
\Delta\ge \frac{1}{n}\sum_{v\in V(G)}d(v)
=
6-\frac{12}{n},
\]
we have $3n(\Delta-6)+36\ge 0$.
Thus the improved denominator is no larger than the old denominator, and so the improved bound is at least as strong.
\end{rmk}

We next record a variant for properly $3$-colored triangulations. In such a triangulation, every facial triangle contains exactly one vertex from each color class. Therefore the degree sum on a face can be bounded using the largest degree in each color class, rather than using $3\Delta(G)$.

\begin{thm}\label{thm:colored-k3-packing}
Let $G$ be a properly $3$-colored simple plane triangulation on $n>3$
vertices, with color classes $V_1,V_2,V_3$. For $i\in\{1,2,3\}$, define $\Delta_i\coloneqq\max\{d_G(v):v\in V_i\}$.
Then
\[
\nu_3(G)\ge
\left\lceil
\frac{2n-4}{\Delta_1+\Delta_2+\Delta_3-5}
\right\rceil .
\]
\end{thm}

\begin{proof}
Every facial triangle of $G$ contains exactly one vertex from each color
class. Hence, if $f=v_1v_2v_3$ is a facial triangle with
$v_i\in V_i$, then $d(v_1)+d(v_2)+d(v_3)-5
\le
\Delta_1+\Delta_2+\Delta_3-5$.
By \Cref{thm:degree-sensitive-k3-packing},
\[
\nu_3(G)
\ge
\left\lceil
\sum_{v_1v_2v_3\in F(G)}
\frac{1}{d(v_1)+d(v_2)+d(v_3)-5}
\right\rceil .
\]
Therefore
\[
\nu_3(G)
\ge
\left\lceil
\frac{|F(G)|}{\Delta_1+\Delta_2+\Delta_3-5}
\right\rceil .
\]
Since $|F(G)|=2n-4$, the result follows.
\end{proof}

\begin{rmk}
The bound in \Cref{thm:colored-k3-packing} strengthens
\Cref{cor:old-delta-k3-packing} whenever $\Delta_1+\Delta_2+\Delta_3<3\Delta(G)$.
Indeed, since $\Delta_i\le \Delta(G)$ for each $i\in\{1,2,3\}$, we have $\Delta_1+\Delta_2+\Delta_3-5\le 3\Delta(G)-5$.
Therefore we obtain
\[
\left\lceil
\frac{2n-4}{\Delta_1+\Delta_2+\Delta_3-5}
\right\rceil
\ge
\left\lceil
\frac{2n-4}{3\Delta(G)-5}
\right\rceil .
\]
At the level of the underlying real-valued fractions, the improvement is
strict whenever
\[
\Delta_1+\Delta_2+\Delta_3<3\Delta(G).
\]
After applying ceilings, however, the two integer lower bounds may coincide.
\end{rmk}

The following example shows that the colored bound can genuinely improve the maximum-degree bound. 
The improvement occurs when the largest degrees are not spread across all three color classes.
\begin{exa}
For every integer $m\ge 3$, there exists a properly $3$-colored triangulation $G$ such that $\Delta_1+\Delta_2+\Delta_3<3\Delta(G)$.
Indeed, let $G$ be the bipyramid over the even cycle $C_{2m}$. 
Thus $V(G)=\{a,b\}\cup \{v_1,\dots,v_{2m}\}$,
where $v_1v_2\cdots v_{2m}v_1$ is a cycle, and each of $a,b$ is adjacent to every $v_i$. Then $G$ is a plane triangulation on $n=2m+2$ vertices.
Since $a$ and $b$ are non-adjacent and the cycle is even, $G$ has a proper $3$-coloring given by
\[
V_1=\{a,b\},\qquad
V_2=\{v_1,v_3,\dots,v_{2m-1}\},\qquad
V_3=\{v_2,v_4,\dots,v_{2m}\}.
\]

Now we have  $d_G(a)=d_G(b)=2m$, and every cycle vertex has degree $4$. 
Hence $\Delta(G)=2m$,
while $\Delta_1=2m$, $\Delta_2=4$,$\Delta_3=4$.
Therefore we have $\Delta_1+\Delta_2+\Delta_3=2m+8$.
Since $m\ge 3$, we obtain $2m+8<6m=3\Delta(G)$.
Thus we get $\Delta_1+\Delta_2+\Delta_3<3\Delta(G)$.

\begin{figure}[H]
\centering

\begin{minipage}[t]{0.48\textwidth}
    \centering
    \begin{tikzpicture}[scale=1.2, every node/.style={font=\small}]
      \coordinate (a) at (0,2.2);
      \coordinate (b) at (0,-2.2);
      \coordinate (v1) at (2,0);
      \coordinate (v2) at (1,1.732);
      \coordinate (v3) at (-1,1.732);
      \coordinate (v4) at (-2,0);
      \coordinate (v5) at (-1,-1.732);
      \coordinate (v6) at (1,-1.732);

      \draw[thick] (v1)--(v2)--(v3)--(v4)--(v5)--(v6)--(v1);

      \foreach \x in {v1,v2,v3,v4,v5,v6}{
        \draw[thick] (a)--(\x);
        \draw[thick] (b)--(\x);
      }

      \fill[blue!25]  (a) circle (2.2pt);
      \fill[blue!25]  (b) circle (2.2pt);

      \fill[red!35]   (v1) circle (2.2pt);
      \fill[green!35] (v2) circle (2.2pt);
      \fill[red!35]   (v3) circle (2.2pt);
      \fill[green!35] (v4) circle (2.2pt);
      \fill[red!35]   (v5) circle (2.2pt);
      \fill[green!35] (v6) circle (2.2pt);

      \node[above]       at (a) {$a$};
      \node[below]       at (b) {$b$};
      \node[right]       at (v1) {$v_1$};
      \node[above right] at (v2) {$v_2$};
      \node[above left]  at (v3) {$v_3$};
      \node[left]        at (v4) {$v_4$};
      \node[below left]  at (v5) {$v_5$};
      \node[below right] at (v6) {$v_6$};
    \end{tikzpicture}

    \medskip
    \small (a) The bipyramid over $C_6$ with proper $3$-coloring:
    {\color{blue!70!black}$V_1=\{a,b\}$},
    {\color{red!70!black}$V_2=\{v_1,v_3,v_5\}$},
    {\color{green!50!black}$V_3=\{v_2,v_4,v_6\}$}.
\end{minipage}
\hfill
\begin{minipage}[t]{0.48\textwidth}
    \centering
\includegraphics[scale=0.8]{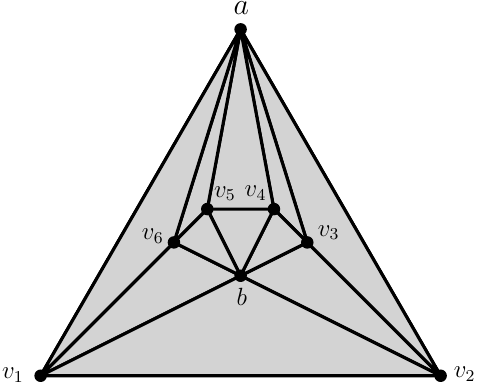}

    \medskip
    \small (b) A plane embedding of the same graph; every face is a triangle.
\end{minipage}

\caption{A small Eulerian triangulation with
$\Delta(G)=6$ and $(\Delta_1,\Delta_2,\Delta_3)=(6,4,4)$, so
$\Delta_1+\Delta_2+\Delta_3=14<18=3\Delta(G)$.}
\label{fig:bipyramid-c6}
\end{figure}
\end{exa}

\subsection{Algorithmic form of the packing bounds}

The lower bounds above are constructive. Let $G$ be a simple plane
triangulation on $n>3$ vertices, and let $H$ be the auxiliary graph whose
vertices are the facial triangles of $G$, with two vertices adjacent whenever
the corresponding faces share a vertex. Every independent set in $H$ gives a
family of pairwise vertex-disjoint facial triangles in $G$, and therefore gives
a $K_3$-packing in $G$.
We can find such a packing in polynomial time as follows. Starting with
$J\coloneqq H$, repeatedly choose a vertex $f$ of minimum degree in the current graph
$J$, add the corresponding facial triangle of $G$ to the packing, and delete
$f$ together with all its neighbours from $J$. The resulting set of chosen
vertices is independent in $H$.
Moreover, this greedy algorithm achieves the Caro--Wei bound used above. Let
\[
\Phi(J)\coloneqq\sum_{u\in V(J)}\frac{1}{d_J(u)+1}.
\]
Suppose the algorithm chooses a vertex $f$ of minimum degree $\delta$ in $J$,
and let $R=\{f\}\cup N_J(f)$.
Then $|R|=\delta+1$, and every vertex of $R$ has degree at least $\delta$.
Hence
\[
\sum_{u\in R}\frac{1}{d_J(u)+1}\le 1.
\]
After deleting $R$, all remaining degrees can only decrease. Thus the
Caro--Wei potential of the remaining graph is at least the contribution, in
$\Phi(J)$, of the vertices outside $R$. Therefore, by induction, the greedy algorithm outputs an independent set of size at least
\[
\sum_{f\in V(H)}\frac{1}{d_H(f)+1}.
\]
Using \Cref{lem:deg_H}, this gives a $K_3$-packing of size at least
\[
\sum_{xyz\in F(G)}
\frac{1}{d_G(x)+d_G(y)+d_G(z)-5}.
\]
Since the output size is an integer, it is at least
\[
\left\lceil
\sum_{xyz\in F(G)}
\frac{1}{d_G(x)+d_G(y)+d_G(z)-5}
\right\rceil .
\]
Consequently, the same algorithm constructively yields all the preceding
Caro--Wei-type lower bounds, such as
\[
\nu_3(G)\ge
\left\lceil \frac{2n-4}{3\Delta(G)-5}\right\rceil .
\]

The running time is polynomial. Indeed, $|V(H)|=|F(G)|=2n-4$, and $H$ can be
constructed by putting a clique on the set of facial triangles incident with
each vertex of $G$. Thus
\[
|E(H)|\le \sum_{v\in V(G)} \binom{d_G(v)}{2}=O(n^2).
\]
Hence the greedy algorithm can be implemented in $O(n^2)$ time.

The computational complexity of maximum triangle packing in plane triangulations
was asked as follows.

Given a plane triangulation $G$ and an integer $k$, decide whether
$\nu_3(G)\ge k$, where $\nu_3(G)$ denotes the maximum number of pairwise
vertex-disjoint triangles in $G$.

Is this problem polynomial-time solvable on plane triangulations, or is it
NP-complete? Is it NP-complete even for $4$-connected plane triangulations?

We answer this question completely.

\begin{thm}
The following problem is NP-complete, even when restricted to
$4$-connected plane triangulations:
given a plane triangulation $G$ and an integer $k$, decide whether $\nu_3(G) \ge k$, where $\nu_3(G)$ denotes the maximum number of pairwise vertex-disjoint triangles in $G$.
\end{thm}

\begin{proof}
The problem is in NP, since a family of at most $|V(G)|/3$ triangles can be checked in polynomial time.

We prove NP-hardness by reducing from \textsc{Independent Set} on $3$-connected cubic planar graphs of girth at least $4$.  
This restricted problem is NP-complete.  Indeed, it is equivalent to the corresponding vertex-cover problem under the transformation $k \mapsto |V(H)|-k$, since $S \subseteq V(H)$ is independent if and only if $V(H)\sm  S$ is a vertex cover.  
The required vertex-cover hardness follows from Uehara's construction \cite{Uehara1996}\footnote{See also \cite[Lemma~3.1 and Remark~3.8]{seo2026complexity} for a detailed account of the $3$-connected cubic planar case and the girth-at-least-$4$ variant.}.

Let $(H,k)$ be such an instance.  
Fix a plane embedding of $H$, and let $G_0 \coloneqq L(H)$ be the line graph of $H$, embedded in the natural way.

For each vertex $a \in V(H)$, the three edges of $H$ incident on $a$ form a triangle in $G_0$; denote this triangle by $\Delta_a$.
We note that every triangle in a line graph either comes from a triangle in the original graph or from three edges incident to a common vertex.
Since $H$ is triangle-free, every triangle of $G_0$ is of the latter form.

In addition, we have $\Delta_a \cap \Delta_b \neq \emptyset$ if and only if $ab \in E(H)$, because then the edge $ab$ of $H$ is a vertex of both corresponding triangles in $L(H)$.  
Therefore vertex-disjoint triangle packings in $G_0$ are in one-to-one correspondence with independent sets in $H$.
Hence we have $\nu_3(G_0) = \alpha(H)$.
Now consider the faces of the natural plane embedding of $G_0$.  They
are of two types.  
First, for every $a \in V(H)$, the triangle $\Delta_a$ is a facial triangle.  
Second, every face $f$ of $H$ gives a face $C_f$ of $G_0$, whose boundary length is the length of $f$.
Since $H$ has girth at least $4$, all these faces $C_f$ have length at least $4$.

We now triangulate every non-triangular face $F$ of $G_0$ by the following gadget.  
Choose two non-consecutive vertices $u_1,u_2$ on the boundary of $F$.  
Add inside $F$ a path
\[
        u_1 v_F w_F x_F y_F u_2
\]
of length $5$, see \Cref{fig:Gadget}.
This splits $F$ into two faces, say $F^1$ and $F^2$.
For each $i \in \{1,2\}$, add a new vertex $z_i^F$ inside $F^i$, and join $z_i^F$ to every vertex on the boundary of $F^i$.  
Let the resulting plane graph be $G$.
All faces of $G$ are triangles, and the construction creates no loops or parallel edges.  
Thus $G$ is a simple plane triangulation.  
Let
\[
        q \coloneqq |\{F : F \text{ is a non-triangular face of } G_0\}|.
\]
\begin{clm}
We claim that $\nu_3(G) = \nu_3(G_0) + 2q$.
\end{clm}
\begin{clmproof}
We fix one gadget inserted into a face $F$, and put
\[
        I_F \coloneqq \{v_F,w_F,x_F,y_F,z_1^F,z_2^F\}.
\]
Every triangle of $G$ that intersects $I_F$ contains at least one of
$z_1^F,z_2^F$.  Suppose, to the contrary, that a triangle $T$ meets
$I_F$ but avoids both $z_1^F$ and $z_2^F$.  Then $T$ contains one of
$v_F,w_F,x_F,y_F$.  In the graph $G-\{z_1^F,z_2^F\}$, the neighbours
of these four vertices are respectively
\[
\begin{aligned}
N_{G-\{z_1^F,z_2^F\}}(v_F) &= \{u_1,w_F\},&
N_{G-\{z_1^F,z_2^F\}}(w_F) &= \{v_F,x_F\},\\
N_{G-\{z_1^F,z_2^F\}}(x_F) &= \{w_F,y_F\},&
N_{G-\{z_1^F,z_2^F\}}(y_F) &= \{x_F,u_2\}.
\end{aligned}
\]
Each pair on the right-hand side is non-adjacent.  Hence no triangle
can meet $\{v_F,w_F,x_F,y_F\}$ while avoiding both $z_1^F$ and $z_2^F$, a contradiction.  Thus every triangle meeting $I_F$ contains $z_1^F$ or $z_2^F$.
It follows that a vertex-disjoint triangle packing can contain at most
two triangles meeting $I_F$, one using $z_1^F$ and one using $z_2^F$.
On the other hand, the two triangles
\[
        z_1^F v_F w_F
        \quad\text{and}\quad
        z_2^F x_F y_F
\]
are vertex-disjoint and use only the vertices of $I_F$.  Therefore every
gadget contributes exactly two triangles independently of the rest of
the graph.

Formally, let $\mathcal{P}$ be any triangle packing in $G$.  For each
non-triangular face $F$ of $G_0$, at most two triangles of
$\mathcal{P}$ meet $I_F$.  Removing all triangles that meet some
$I_F$, the remaining triangles use only the vertices of $G_0$.  Since no
edges between old vertices were added, those remaining triangles are
triangles of $G_0$.  
Hence we obtain  $|\mathcal{P}| \le \nu_3(G_0) + 2q$.
Conversely, take a maximum triangle packing of $G_0$, and add, for each gadget, the two internal triangles $z_1^F v_F w_F$ and $z_2^F x_F y_F$.
These added triangles are pairwise vertex-disjoint and avoid all old vertices.  Hence we have $\nu_3(G) \ge \nu_3(G_0) + 2q$.
Thus we obtain $\nu_3(G) = \nu_3(G_0) + 2q = \alpha(H) + 2q$.
Consequently, we have $\alpha(H) \ge k$ if and only if $\nu_3(G) \ge k + 2q$. 
\end{clmproof}

It remains to verify that $G$ is $4$-connected.

\begin{clm}
We claim that $G$ is $4$-connected.
\end{clm}
\begin{clmproof}
We show that every
triangle of $G$ is facial.  Since $G$ is a simple plane triangulation,
this is equivalent to saying that $G$ has no separating triangle, and
hence $G$ is $4$-connected.

First consider a triangle $T$ of $G$ using no new gadget vertex.  Then
$T$ is a triangle of $G_0=L(H)$, and, as observed above, every
triangle of $G_0$ is one of the facial triangles $\Delta_a$.

Now suppose $T$ uses a vertex from some gadget inserted in a face
$F$.  It cannot use vertices from two different gadgets, since no edge
joins new vertices belonging to different gadgets.  If $T$ uses one of
$v_F,w_F,x_F,y_F$, then, by the argument above, $T$ also uses
$z_1^F$ or $z_2^F$.  Hence $T$ has the form $z_i^F ab$,
where $i \in \{1,2\}$ and $a,b$ are adjacent vertices on the boundary
of $F^i$.

The boundary of $F^i$ has no chords.  Indeed, the old part of this
boundary is an arc of the non-triangular face $F=C_f$ of $G_0=L(H)$,
where $f$ is a face of $H$.  Two vertices on $C_f$ are adjacent in $L(H)$ only when the corresponding two edges of $H$ are consecutive on the facial cycle bounding $f$.  
Thus $C_f$ has no chords in $G_0$.  
The newly added path $u_1 v_F w_F x_F y_F u_2$
also has no chords, no internal vertex of this path is adjacent to an old boundary vertex except through the prescribed path edge, and $u_1u_2 \notin E(G_0)$ because $u_1,u_2$ were chosen non-consecutive on $C_f$.  
Therefore $a$ and $b$ are consecutive on the boundary of
$F^i$, so $z_i^F ab$ is one of the facial triangles created by the fan from $z_i^F$.

Therefore every triangle of $G$ is facial.  Thus $G$ is a
$4$-connected plane triangulation. 
\end{clmproof}
The reduction maps $(H,k)$ to
$(G,k+2q)$ in polynomial time, so the problem is NP-hard even on $4$-connected plane triangulations.  Since the problem is in NP, it is NP-complete.
\end{proof}

\begin{figure}
    \centering
    \includegraphics[scale=0.7]{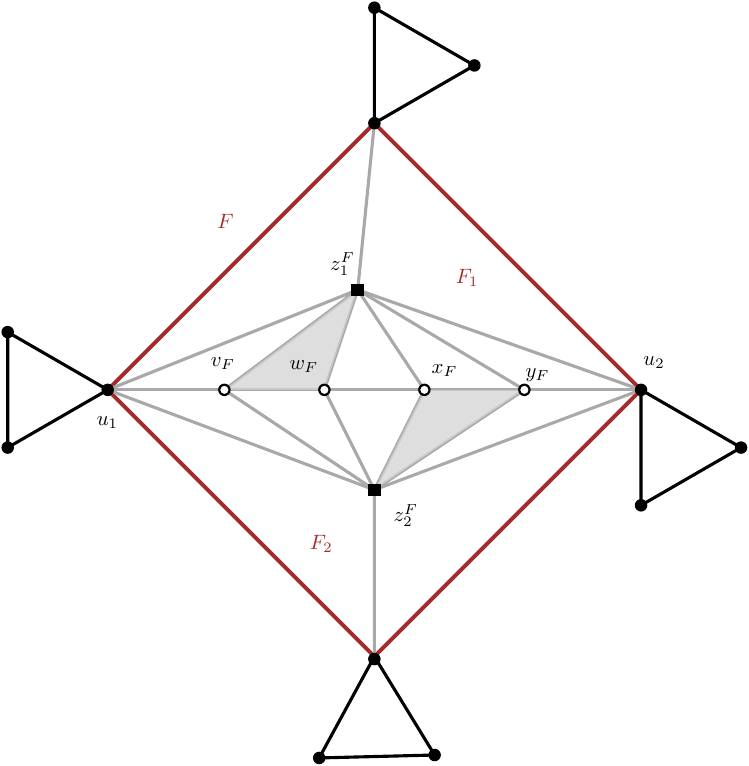}
\caption{The face gadget used to triangulate a non-triangular face $F$ of $G_0$. 
Two non-consecutive boundary vertices $u_1$ and $u_2$ are connected by the path
$u_1v_Fw_Fx_Fy_Fu_2$, splitting $F$ into two regions $F_1$ and $F_2$. 
For each $i\in\{1,2\}$, a vertex $z_i^F$ is added inside $F_i$ and joined to every
vertex on the boundary of $F_i$. The gadget triangulates $F$ and contributes the
two internal vertex-disjoint triangles $z_1^Fv_Fw_F$ and $z_2^Fx_Fy_F$.}
\label{fig:face-gadget}    \label{fig:Gadget}
\end{figure}
\begin{rmk}
If one starts instead from the planar vertex-disjoint-triangles
hardness theorem of Guruswami et al.~
\cite{GuruswamiRanganChangChangWong1998}, the same face gadget already
proves NP-hardness for plane triangulations.  The extra line-graph
starting point above is what guarantees that the final triangulation
has no separating triangles, and therefore is $4$-connected.
\end{rmk}

\section{\texorpdfstring{$K_3$-factors}{k3-factors}}

\begin{defn}
A \defin{triangle factor} or \defin{$K_3$-factor} is a packing that covers all vertices.
In other words, a triangle factor in a graph 
$G$ is a spanning subgraph whose connected components are all triangles.
\end{defn}

\begin{lem}[\citet{CorradiHajnal1963}] 
Any $n$-vertex graph $G$ with $n \in 3\mathbb{N}$ and $\delta(G) \ge \frac{2n}{3}$ contains a triangle factor.   
\end{lem}

Let $G$ be a $4$-connected plane triangulation, and let $C$ be a hamiltonian cycle of $G$.
Since $C$ passes through every vertex of $G$, the embedding of $C$ separates the graph into exactly two regions: the interior of $C$ and the exterior of $C$. Every edge of $G$ is therefore either an edge of $C$, an edge drawn inside $C$, or an edge drawn outside $C$.
Let $E^+$ be the set of edges drawn on one side of $C$, and let $E^-$ be the set of edges drawn on the other side of $C$. Define
\[
G^+ = (V(G),\, E(C) \cup E^+),
\qquad
G^- = (V(G),\, E(C) \cup E^-).
\]
These are exactly the two outerplanar graphs determined by the hamiltonian cycle $C$. In fact, because $G$ is a triangulation, both $G^+$ and $G^-$ are maximal outerplanar graphs: in the embeddings with $C$ as the outer boundary, all vertices lie on the outer face, and every bounded face is a triangle.

\begin{obs}
If $G$ is a maximal outerplanar graph, then its weak dual is a tree.
\end{obs}

Now we need a couple of notations.
Let $G$ be a $4$-connected triangulation on $n$ vertices. Let $C=v_1v_2\cdots v_nv_1$ be a hamiltonian cycle of $G$, and let $G^+$ and $G^-$
be the two maximal outerplanar graphs determined by $C$, as described above.

For a plane outerplanar graph $O$, let $F_b(O)$ denote the set of bounded
faces of $O$. Thus the faces in $F_b(G^+)$ and $F_b(G^-)$ are precisely the
triangular faces of $G$ lying on the two sides of $C$.

For $\sigma\in\{+,-\}$, let $T^\sigma$ be the weak dual of $G^\sigma$,
that is, the graph whose vertices are the bounded faces of $G^\sigma$, with
two vertices adjacent if the corresponding bounded faces share an edge, see \Cref{fig:4conn-three-figures}.

\begin{defn}
For each $j\in\{1,\dots,n\}$ and $\sigma\in\{+,-\}$, let $P_j^\sigma$
denote the subgraph of $T^\sigma$ induced by the bounded faces of
$G^\sigma$ incident with $v_j$. Since $G^\sigma$ is maximal outerplanar, $P_j^\sigma$ is a path; call it the \defin{face-path} of $v_j$ in $T^\sigma$, see \Cref{fig:4conn-three-figures}.
\end{defn}

\begin{thm}\label{thm:4-conn}
Let $G$ be a $4$-connected plane triangulation.
Then the following are equivalent.

\begin{enumerate}
\item $G$ has a triangle factor.

\item $G^*$ has an independent set
$X\subseteq V(G^*)$ such that every face of $G^*$ contains exactly one vertex
of $X$.

\item There exist sets of bounded faces
$\mathcal T^+\subseteq F_b(G^+)$ and $\mathcal T^-\subseteq F_b(G^-)$ such that
every vertex $v_j$ of $C$ belongs to exactly one face of
$\mathcal T^+\cup\mathcal T^-$.

\item There exist sets $X^+\subseteq V(T^+)$ and $X^-\subseteq V(T^-)$ such
that, for every $j\in\{1,\dots,n\}$,
\[
|X^+\cap V(P_j^+)|+|X^-\cap V(P_j^-)|=1.
\]
\end{enumerate}
\end{thm}

\begin{proof}
Since $G$ is $4$-connected, every triangle of $G$ is facial. Hence a
triangle factor of $G$ is exactly a set of pairwise vertex-disjoint faces of
$G$ covering $V(G)$.

$(1)\Rightarrow(2)$.
Let $\mathcal T$ be a triangle factor of $G$, and let
$X=\{f^*:f\in\mathcal T\}\subseteq V(G^*)$.
If two vertices of $X$ were adjacent in $G^*$, then the corresponding faces
of $G$ would share an edge, hence two vertices, contradicting the
vertex-disjointness of $\mathcal T$. So $X$ is independent.
For each $v\in V(G)$, the faces of $G$ incident with $v$ correspond to
the vertices on the dual face of $G^*$ associated with $v$. Since
$\mathcal T$ covers $v$ exactly once, so the dual face contains exactly one
vertex of $X$. Thus $X$ is an independent set.

$(2)\Rightarrow(1)$.
Let $X\subseteq V(G^*)$ be independent and let $\mathcal T=\{f:f^*\in X\}$.
If $f,g\in\mathcal T$ shared a vertex $v$, then both $f^*$ and $g^*$
would lie on the dual face corresponding to $v$, contradicting interdependence. 
Hence the faces in $\mathcal T$ are pairwise vertex-disjoint.
Again by face-exactness, each $v\in V(G)$ is incident with exactly one face
of $\mathcal T$. Therefore $\mathcal T$ is a triangle factor.

$(2)\Leftrightarrow(3)$.
The bounded faces of $G^+$ together with the bounded faces of $G^-$ are
precisely the faces of the triangulation $G$; the outer faces of $G^+$ and
$G^-$ are not faces of $G$ and are not included. Therefore a set
$X\subseteq V(G^*)$ is the same thing as a pair of sets
$(\mathcal T^+,\mathcal T^-)$ with
\[
\mathcal T^\sigma\subseteq F_b(G^\sigma).
\]
For a fixed $j$, the dual face of $G^*$ corresponding to $v_j$ consists
exactly of the faces of $G$ incident with $v_j$, equivalently the bounded
faces of $G^+$ and $G^-$ incident with $v_j$.
Hence that dual face contains exactly one vertex of $X$ if and only if
$v_j$ belongs to exactly one face of $\mathcal T^+\cup\mathcal T^-$.
So (2) and (3) are equivalent.

$(3)\Leftrightarrow(4)$.
For each $\sigma\in\{+,-\}$, vertices of $T^\sigma$ are bounded faces of
$G^\sigma$. Thus a set of bounded faces
$\mathcal T^\sigma\subseteq F_b(G^\sigma)$ corresponds uniquely to a vertex
set $X^\sigma\subseteq V(T^\sigma)$.
By definition of $P_j^\sigma$, a bounded face of $G^\sigma$ contains $v_j$
if and only if the corresponding vertex of $T^\sigma$ lies on $P_j^\sigma$.
Therefore $v_j$ belongs to exactly one face of
$\mathcal T^+\cup\mathcal T^-$ if and only if
\[
|X^+\cap V(P_j^+)|+|X^-\cap V(P_j^-)|=1.
\]
So (3) and (4) are equivalent.
\end{proof}

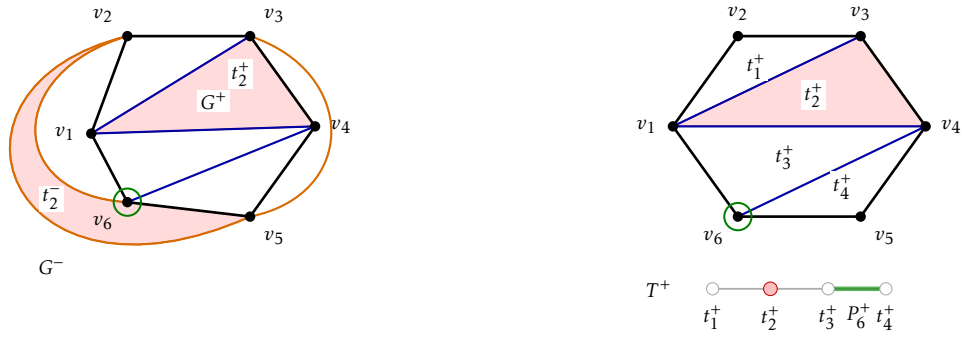
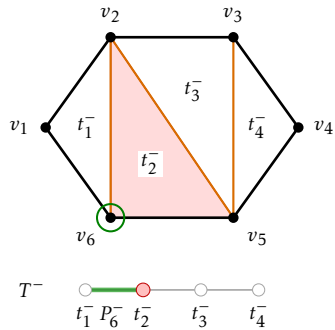
\begin{figure}[H]
\centering
\tikzset{
  every node/.style={font=\small},
  vtx/.style={circle,fill=black,inner sep=1.45pt},
  markedv/.style={circle,draw=green!50!black,line width=.8pt,inner sep=3.8pt},
  ham/.style={black,line width=1.05pt},
  plusedge/.style={blue!65!black,line width=.85pt},
  minusedge/.style={orange!85!black,line width=.85pt},
  selectedface/.style={fill=red!25,draw=red!75!black,line width=.55pt,fill opacity=.55},
  flabel/.style={font=\scriptsize,fill=white,inner sep=1pt},
  dualedge/.style={gray!65,line width=.7pt},
  pathdual/.style={green!55!black,line width=1.8pt,opacity=.65},
  dualnode/.style={circle,draw=gray!70,fill=white,inner sep=1.7pt},
  selnode/.style={circle,draw=red!75!black,fill=red!25,inner sep=1.9pt},
  paneltitle/.style={font=\small\bfseries,fill=white,inner sep=1pt}
}

\newcommand{\SamePanelBox}{%
  \path[use as bounding box] (-3.75,-3.05) rectangle (3.75,1.90);
}

\begin{subfigure}[t]{0.47\textwidth}
\centering
\resizebox{\linewidth}{!}{%
\begin{tikzpicture}[x=1cm,y=1cm]
  \SamePanelBox

  \coordinate (v1) at (-1.35,-0.1);
  \coordinate (v2) at (-.85,1.25);
  \coordinate (v3) at (.85,1.25);
  \coordinate (v4) at (1.75,0);
  \coordinate (v5) at (.85,-1.25);
  \coordinate (v6) at (-.85,-1.05);

  \coordinate (a1) at (-2.55,.75);
  \coordinate (a2) at (-2.55,-.95);
  \coordinate (b1) at (-3.75,.55);
  \coordinate (b2) at (-2.52,-2.75);
  \coordinate (c1) at (2.35,.75);
  \coordinate (c2) at (2.35,-.95);

  \filldraw[selectedface] (v1)--(v3)--(v4)--cycle;
  \filldraw[selectedface]
    (v2) .. controls (b1) and (b2) .. (v5)
         -- (v6)
         .. controls (a2) and (a1) .. (v2)
         -- cycle;

  \draw[ham] (v1)--(v2)--(v3)--(v4)--(v5)--(v6)--cycle;

  \draw[plusedge] (v1)--(v3) (v1)--(v4) (v4)--(v6);

  \draw[minusedge]
    (v2) .. controls (a1) and (a2) .. (v6)
    (v2) .. controls (b1) and (b2) .. (v5)
    (v3) .. controls (c1) and (c2) .. (v5);

  \node[vtx,label={[font=\scriptsize]left:$v_1$}] at (v1) {};
  \node[vtx,label={[font=\scriptsize]above left:$v_2$}] at (v2) {};
  \node[vtx,label={[font=\scriptsize]above right:$v_3$}] at (v3) {};
  \node[vtx,label={[font=\scriptsize]right:$v_4$}] at (v4) {};
  \node[vtx,label={[font=\scriptsize]below right:$v_5$}] at (v5) {};
  \node[vtx,label={[font=\scriptsize]below left:$v_6$}] at (v6) {};
  \node[markedv] at (v6) {};

  \node[flabel] at (.35,.35) {$G^+$};
  \node[flabel] at (-1.90,-1.95) {$G^-$};
  \node[flabel] at (.72,.72) {$t_2^+$};
  \node[flabel] at (-1.9,-1) {$t_2^-$};
\end{tikzpicture}%
}
\caption{The triangulation $G$ with hamiltonian cycle $C$}
\end{subfigure}
\hfill
\begin{subfigure}[t]{0.47\textwidth}
\centering
\resizebox{\linewidth}{!}{%
\begin{tikzpicture}[x=1cm,y=1cm]
  \SamePanelBox

  \coordinate (v1) at (-1.75,0);
  \coordinate (v2) at (-.85,1.25);
  \coordinate (v3) at (.85,1.25);
  \coordinate (v4) at (1.75,0);
  \coordinate (v5) at (.85,-1.25);
  \coordinate (v6) at (-.85,-1.25);

  \filldraw[selectedface] (v1)--(v3)--(v4)--cycle;

  \draw[ham] (v1)--(v2)--(v3)--(v4)--(v5)--(v6)--cycle;
  \draw[plusedge] (v1)--(v3) (v1)--(v4) (v4)--(v6);

  \node[flabel] at (-.60,.83) {$t^+_1$};
  \node[flabel] at (.20,.43) {$t^+_2$};
  \node[flabel] at (-.20,-.43) {$t^+_3$};
  \node[flabel] at (.60,-.83) {$t^+_4$};

  \node[vtx,label={[font=\scriptsize]left:$v_1$}] at (v1) {};
  \node[vtx,label={[font=\scriptsize]above:$v_2$}] at (v2) {};
  \node[vtx,label={[font=\scriptsize]above:$v_3$}] at (v3) {};
  \node[vtx,label={[font=\scriptsize]right:$v_4$}] at (v4) {};
  \node[vtx,label={[font=\scriptsize]below right:$v_5$}] at (v5) {};
  \node[vtx,label={[font=\scriptsize]below left:$v_6$}] at (v6) {};
  \node[markedv] at (v6) {};

  \node[flabel] at (-1.95,-2.25) {$T^+$};

  \coordinate (p1) at (-1.20,-2.25);
  \coordinate (p2) at (-.40,-2.25);
  \coordinate (p3) at (.40,-2.25);
  \coordinate (p4) at (1.20,-2.25);

  \draw[dualedge] (p1)--(p2)--(p3)--(p4);
  \draw[pathdual] (p3)--(p4);

  \node[dualnode,label={[font=\scriptsize]below:$t^+_1$}] at (p1) {};
  \node[selnode,label={[font=\scriptsize]below:$t^+_2$}] at (p2) {};
  \node[dualnode,label={[font=\scriptsize]below:$t^+_3$}] at (p3) {};
  \node[dualnode,label={[font=\scriptsize]below:$t^+_4$}] at (p4) {};

  \node[flabel] at (.83,-2.62) {$P^+_6$};
\end{tikzpicture}%
}
\caption{The inside graph $G^+$ and its weak dual $T^+$}
\end{subfigure}

\vspace{1.2ex}

\begin{subfigure}[t]{0.47\textwidth}
\centering
\resizebox{\linewidth}{!}{%
\begin{tikzpicture}[x=1cm,y=1cm]
  \SamePanelBox

  \coordinate (v1) at (-1.75,0);
  \coordinate (v2) at (-.85,1.25);
  \coordinate (v3) at (.85,1.25);
  \coordinate (v4) at (1.75,0);
  \coordinate (v5) at (.85,-1.25);
  \coordinate (v6) at (-.85,-1.25);

  \filldraw[selectedface] (v2)--(v5)--(v6)--cycle;

  \draw[ham] (v1)--(v2)--(v3)--(v4)--(v5)--(v6)--cycle;
  \draw[minusedge] (v2)--(v6) (v2)--(v5) (v3)--(v5);

  \node[flabel] at (-1.18,0) {$t^-_1$};
  \node[flabel] at (-.30,-.50) {$t^-_2$};
  \node[flabel] at (.30,.50) {$t^-_3$};
  \node[flabel] at (1.18,0) {$t^-_4$};

  \node[vtx,label={[font=\scriptsize]left:$v_1$}] at (v1) {};
  \node[vtx,label={[font=\scriptsize]above:$v_2$}] at (v2) {};
  \node[vtx,label={[font=\scriptsize]above:$v_3$}] at (v3) {};
  \node[vtx,label={[font=\scriptsize]right:$v_4$}] at (v4) {};
  \node[vtx,label={[font=\scriptsize]below right:$v_5$}] at (v5) {};
  \node[vtx,label={[font=\scriptsize]below left:$v_6$}] at (v6) {};
  \node[markedv] at (v6) {};
References
  \node[flabel] at (-1.95,-2.25) {$T^-$};

  \coordinate (m1) at (-1.20,-2.25);
  \coordinate (m2) at (-.40,-2.25);
  \coordinate (m3) at (.40,-2.25);
  \coordinate (m4) at (1.20,-2.25);

  \draw[dualedge] (m1)--(m2)--(m3)--(m4);
  \draw[pathdual] (m1)--(m2);

  \node[dualnode,label={[font=\scriptsize]below:$t^-_1$}] at (m1) {};
  \node[selnode,label={[font=\scriptsize]below:$t^-_2$}] at (m2) {};
  \node[dualnode,label={[font=\scriptsize]below:$t^-_3$}] at (m3) {};
  \node[dualnode,label={[font=\scriptsize]below:$t^-_4$}] at (m4) {};

  \node[flabel] at (-.83,-2.62) {$P^-_6$};
\end{tikzpicture}%
}
\caption{The outside graph $G^-$ and its weak dual $T^-$}
\end{subfigure}

\caption{Figure (a) shows the full triangulation $G$ with hamiltonian cycle
$C=v_1v_2\cdots v_6v_1$. Figures (b) and (c) show the two maximal
outerplanar graphs $G^+$ and $G^-$ determined by $C$, together with their
weak duals $T^+$ and $T^-$. The red triangles/nodes indicate the selected
faces $t_2^+$ and $t_2^-$, while the green paths indicate $P_6^+$ and
$P_6^-$.}
\label{fig:4conn-three-figures}
\end{figure}

\begin{cor}\label{lem:demand-4conn}
Let $G$ be a $4$-connected plane triangulation. 
Let $I\colon V(G)\to\{0,1\}$.
Then the following are equivalent.

\begin{enumerate}
\item There exists a family $\mathcal T$ of pairwise vertex-disjoint triangles
of $G$ such that every vertex $v$ of $G$ belongs to exactly $I(v)$ triangles of
$\mathcal T$.

\item There exist sets of bounded faces
$\mathcal T^+\subseteq F_b(G^+)$ and
$\mathcal T^-\subseteq F_b(G^-)$ such that every vertex $v_j$ of $C$ belongs
to exactly $I(v_j)$ faces of $\mathcal T^+\cup\mathcal T^-$.

\item There exist sets $X^+\subseteq V(T^+)$ and $X^-\subseteq V(T^-)$
such that, for every $j\in\{1,\dots,n\}$,
\[
|X^+\cap V(P_j^+)|+|X^-\cap V(P_j^-)|=I(v_j).
\]
\end{enumerate}
\end{cor}

\begin{proof}
Since $G$ is $4$-connected, every triangle of $G$ is facial. Hence a family of pairwise vertex-disjoint triangles of $G$ is the same thing as a set of pairwise vertex-disjoint faces of $G$. The bounded faces of $G^+$ together with the bounded faces of $G^-$ are precisely the faces of $G$.
The proof of \Cref{thm:4-conn} applies verbatim after replacing the constant right-hand side $1$ by the prescribed demand $I(v_j)$ at each vertex $v_j$.
Thus the three conditions are equivalent.
\end{proof}

The approximation is obtained by a standard splitting argument: any optimal packing splits into triangles lying inside $C$ and triangles lying outside $C$, so one of the two sides contains at least half of the optimum.

\begin{thm}\label{cor:4conn-two-approx}
Let $G$ be a $4$-connected plane triangulation. Then the maximum triangle packing problem in $G$ has a $2$-approximation.
\end{thm}

\begin{proof}
Let $G^+$ and $G^-$ be the two maximal outerplanar graphs determined by
$C$. For $\sigma\in\{+,-\}$, define an auxiliary graph $H^\sigma$ whose
vertices are the bounded faces of $G^\sigma$, with two vertices adjacent
whenever the corresponding faces share a vertex of $G$.
An independent set  in $H^\sigma$ is exactly a family of pairwise vertex-disjoint
triangular faces of $G^\sigma$. Thus it gives a triangle packing in $G$ using
only triangles on the $\sigma$-side of $C$. Let $S^\sigma$ be a maximum stable
set in $H^\sigma$, and return the larger of $S^+$ and $S^-$.
Let $\mathcal O$ be an optimal triangle packing in $G$. Since $G$ is
$4$-connected, every triangle of $G$ is facial. Hence every triangle in
$\mathcal O$ lies on exactly one side of $C$. 
Write $\mathcal O=\mathcal O^+\cup \mathcal O^-$,
where $\mathcal O^\sigma$ is the set of triangles of $\mathcal O$ lying in
$G^\sigma$. Since $\mathcal O^\sigma$ is a feasible independent set  in
$H^\sigma$, we have $|S^\sigma|\ge |\mathcal O^\sigma|$ for each $\sigma\in\{+,-\}$.
Therefore
\[
\max\{|S^+|,|S^-|\}
\ge
\frac{|S^+|+|S^-|}{2}
\ge
\frac{|\mathcal O^+|+|\mathcal O^-|}{2}
=
\frac{|\mathcal O|}{2}.
\]
Thus the larger of the two one-sided solutions has size at least half of the
optimum, giving a $2$-approximation.
\end{proof}

\begin{rmk}
The above $2$-approximation is not the best possible from the general planar
algorithmic point of view. Since triangle packing is the maximum
$H$-matching problem for the fixed graph $H=K_3$, Baker's layering technique
for planar graphs gives a PTAS~\cite{Baker1994}. In particular, for every
fixed $\varepsilon>0$, there is a polynomial-time algorithm that, given a
plane triangulation $G$, finds a triangle packing of size at least $(1-\varepsilon)\nu_3(G)$.
This applies in particular to $4$-connected plane triangulations.
\end{rmk}
\section{\texorpdfstring{$P_2\cup P_1$-packings}{p1p2-packing}}

In this section, we study packings by induced copies of $P_2\cup P_1$.
A $3$-vertex set induces $P_2\cup P_1$ in a triangulation $T$ precisely when it
induces a $P_3$ in $\overline T$. We exploit this complementary viewpoint
indirectly, using a large matching in $T$ and a planar bipartite obstruction.

\begin{defn}
Let $G$ be a graph. A \defin{$(P_2\cup P_1)$-packing} in $G$ is a family of pairwise vertex-disjoint induced subgraphs of $G$, each isomorphic to $P_2\cup P_1$. We write $\lambda_{P_2\cup P_1}(G)$
for the maximum size of such a packing.
\end{defn}

\begin{thm}\label{thm:p2p1-packing}
Let $T$ be a plane triangulation on $n$ vertices. Then
$\lambda_{P_2\cup P_1}(T)\ge \left\lfloor \frac n3\right\rfloor-2$.
\end{thm}

\begin{proof}
We put $q=\left\lfloor n/3\right\rfloor$.
If $q\le 2$, the assertion is trivial, so assume $q\ge 3$.
By the theorem of Nishizeki and Baybars~\cite{NishizekiBaybars1979}, every
connected planar graph with minimum degree at least $3$ has a matching of size
at least $\lfloor n/3\rfloor$. Since $T$ is a plane triangulation and
$q\ge 3$, the graph $T$ has a matching $M=\{e_1,\dots,e_q\}$
of size $q$.
Let $U\coloneqq V(T)\sm V(M)$. Since $n=3q+r$, where $r\in\{0,1,2\}$, we have $|U|=n-2q=q+r\ge q$.
We define a bipartite graph $C$ with parts $M$ and $U$. For an edge
$e=xy\in M$ and a vertex $u\in U$, join $e$ to $u$ in $C$ if $u$ is adjacent
in $T$ to $x$ or to $y$. The graph $C$ is planar: contract every edge of $M$ in the planar embedding of $T$, and then keep only the edges between the contracted vertices corresponding to $M$ and the vertices of $U$.
We now construct the packing one copy at a time. Suppose that after
$i<q-2$ steps, the unused sets are $M_i\subseteq M$ and $U_i\subseteq U$.
Then $|M_i|=q-i\ge 3$ and $|U_i|\ge q-i\ge 3$.
Since $C[M_i\cup U_i]$ is planar, it is not complete bipartite; otherwise it
would contain a copy of $K_{3,3}$. Hence there are $e=xy\in M_i$ and
$u\in U_i$ that are non-adjacent in $C$.
By the definition of $C$, the vertex $u$ is adjacent to neither $x$ nor $y$ in $T$. 
Since $xy\in E(T)$, the set $\{x,y,u\}$ induces a copy of
$P_2\cup P_1$ in $T$. 
We add this copy to the packing, delete $e$ from the
unused matching edges, delete $u$ from the unused vertices, and continue.
After $q-2$ steps, this produces $q-2$ pairwise vertex-disjoint induced copies of $P_2\cup P_1$. Therefore
\[
\lambda_{P_2\cup P_1}(T)\ge q-2
=
\left\lfloor \frac n3\right\rfloor-2.\qedhere
\]
\end{proof}
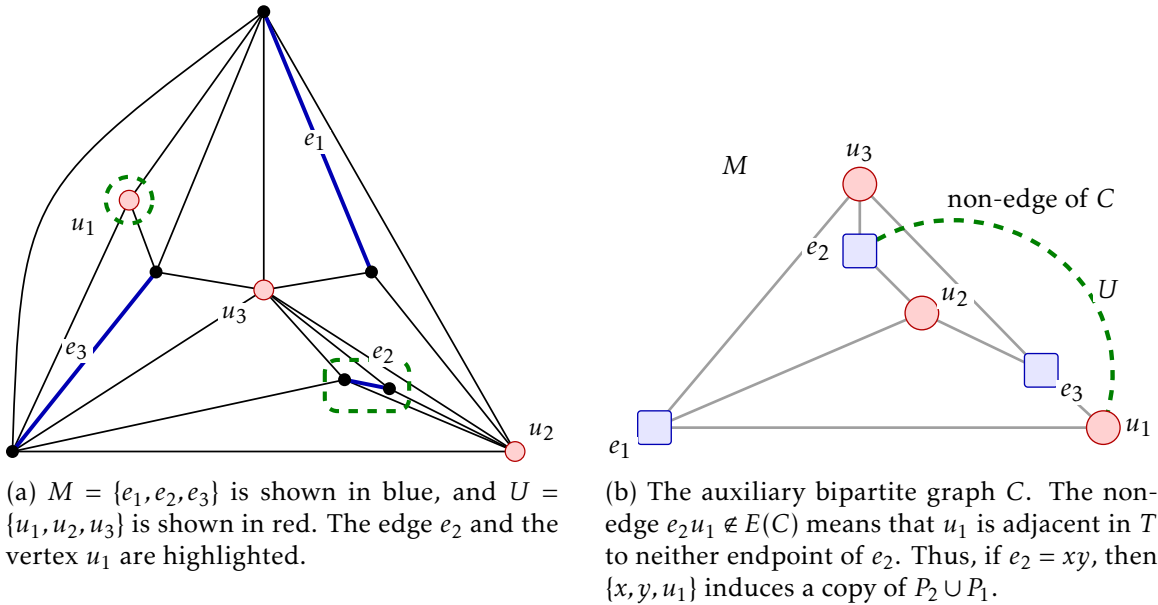
\begin{figure}[H]
\centering
\tikzset{
  vtx/.style={circle,fill=black,inner sep=1.45pt},
  uvtx/.style={circle,draw=red!70!black,fill=red!18,inner sep=2.2pt},
  edgeT/.style={black,line width=.55pt},
  medge/.style={blue!70!black,line width=1.25pt},
  auxedge/.style={gray!75,line width=.8pt},
  chosen/.style={green!50!black,line width=1.3pt},
  boxlab/.style={font=\small\bfseries},
  smlab/.style={font=\scriptsize,fill=white,inner sep=1pt},
  enode/.style={rectangle,draw=blue!70!black,fill=blue!10,
                rounded corners=1pt,minimum width=10pt,
                minimum height=10pt,inner sep=0pt},
  unode/.style={circle,draw=red!70!black,fill=red!18,
                minimum size=10pt,inner sep=0pt}
}

\begin{subfigure}[t]{0.48\textwidth}
\centering
\resizebox{\linewidth}{!}{%
\begin{tikzpicture}[x=1cm,y=1cm]
  \coordinate (A) at (0,3.1);
  \coordinate (B) at (-2.8,-1.8);
  \coordinate (C) at (2.8,-1.8);

  \coordinate (D) at (0,0);
  \coordinate (P) at (-1.2,0.2);
  \coordinate (Q) at (0.9,-1.0);
  \coordinate (R) at (1.2,0.2);
  \coordinate (Uone) at (-1.5,1);
  \coordinate (S) at (1.4,-1.1);

  \draw[edgeT] (B)--(C)--(A);
  \draw[edgeT] (A) .. controls (-2.7,1.1) .. (B);

  \draw[edgeT] (A)--(D) (B)--(D) (C)--(D);

  \draw[edgeT] (P)--(A) (P)--(B) (P)--(D);

  \draw[edgeT] (Q)--(B) (Q)--(C) (Q)--(D);

  \draw[edgeT] (R)--(C) (R)--(A) (R)--(D);

  \draw[edgeT] (Uone)--(A) (Uone)--(B) (Uone)--(P);

  \draw[edgeT] (S)--(C) (S)--(D) (S)--(Q);

  \draw[medge] (A)--(R);   
  \draw[medge] (Q)--(S);   
  \draw[medge] (B)--(P);   

  \draw[chosen,dashed,rounded corners] (0.68,-0.78) rectangle (1.62,-1.35);
  \draw[chosen,dashed] (Uone) circle (.26);

  \foreach \p in {A,B,P,Q,R,S}
    \node[vtx] at (\p) {};

  \foreach \p in {Uone,C,D}
    \node[uvtx] at (\p) {};

  \node[smlab] at (.58,1.65) {$e_1$};
  \node[smlab] at (1.34,-.72) {$e_2$};
  \node[smlab] at (-2.08,-.68) {$e_3$};

  \node[smlab] at (-2,0.7) {$u_1$};
  \node[smlab] at (3.10,-1.58) {$u_2$};
  \node[smlab] at (-.32,-.28) {$u_3$};
\end{tikzpicture}%
}
\caption{$M=\{e_1,e_2,e_3\}$ is shown in blue, and
$U=\{u_1,u_2,u_3\}$ is shown in red. The edge $e_2$ and the vertex $u_1$
are highlighted.}
\end{subfigure}
\hfill
\begin{subfigure}[t]{0.48\textwidth}
\centering
\resizebox{\linewidth}{!}{%
\begin{tikzpicture}[x=1cm,y=1cm]
  \coordinate (E1) at (-2.6,-1.0);
  \coordinate (E2) at (-0.45,0.85);
  \coordinate (E3) at (1.45,-0.4);

  \coordinate (U1) at (2.1,-1.0);
  \coordinate (U2) at (0.2,0.2);
  \coordinate (U3) at (-0.45,1.55);

  \draw[auxedge] (E1)--(U1);
  \draw[auxedge] (E1)--(U2);
  \draw[auxedge] (E1)--(U3);

  \draw[auxedge] (E2)--(U2);
  \draw[auxedge] (E2)--(U3);

  \draw[auxedge] (E3)--(U1);
  \draw[auxedge] (E3)--(U2);
  \draw[auxedge] (E3)--(U3);

  \draw[chosen,dashed]
    (E2) .. controls (.75,1.75) and (2.65,.65) .. (U1);
  \node[smlab] at (1.35,1.38) {non-edge of $C$};

  \node[enode] at (E1) {};
  \node[enode] at (E2) {};
  \node[enode] at (E3) {};

  \node[unode] at (U1) {};
  \node[unode] at (U2) {};
  \node[unode] at (U3) {};

  \node[smlab] at (-2.95,-1.23) {$e_1$};
  \node[smlab] at (-.88,.85) {$e_2$};
  \node[smlab] at (1.78,-.63) {$e_3$};

  \node[smlab] at (2.48,-1.0) {$u_1$};
  \node[smlab] at (.55,.37) {$u_2$};
  \node[smlab] at (-.45,1.88) {$u_3$};

  \node[smlab] at (-1.75,1.75) {$M$};
  \node[smlab] at (2.15,.45) {$U$};
\end{tikzpicture}%
}
\caption{The auxiliary bipartite graph $C$. The non-edge
$e_2u_1\notin E(C)$ means that $u_1$ is adjacent in $T$ to neither endpoint
of $e_2$. Thus, if $e_2=xy$, then $\{x,y,u_1\}$ induces a copy of
$P_2\cup P_1$.}
\end{subfigure}

\caption{Illustration of the construction in the proof of
\Cref{thm:p2p1-packing}. The left figure is a triangulation. The right figure
shows the bipartite graph $C$ with parts $M$ and $U$. An edge $eu$ is present
in $C$ whenever the unmatched vertex $u$ is adjacent in $T$ to at least one
endpoint of the matching edge $e$. Hence a non-edge $eu$ in $C$ gives an
induced copy of $P_2\cup P_1$ in $T$.}
\label{fig:p2p1-schematic}
\end{figure}

\begin{cor}\label{cor:p2p1-algorithm}
A packing of size at least $\lfloor n/3\rfloor-2$ can be found in
$O(n^2)$ time.
\end{cor}

\begin{proof}
Let  $q=\left\lfloor n/3 \right\rfloor$.
If $q\le 2$, then the empty packing already has size at least $q-2$, so assume $q\ge 3$.
Since $T$ has $3n-6$ edges, a maximum matching in $T$ can be found in
$O(n^{3/2})$ time by the Micali--Vazirani algorithm~\cite{MicaliVazirani1980}.
By the theorem of Nishizeki and Baybars~\cite{NishizekiBaybars1979}, this
matching has size at least $q$. Choose $q$ edges from it and call the resulting matching $M$. 
Let $U\coloneqq V(T)\sm V(M)$.
Build an adjacency matrix for $T$ in $O(n^2)$ time. Then form all
admissible pairs $(e,u)\in M\times U$ such that, writing $e=xy$, the
vertex $u$ is adjacent to neither $x$ nor $y$. This takes $O(|M||U|)=O(n^2)$ time.

Scan these admissible pairs, selecting a pair whenever neither the
edge of $M$ nor the vertex of $U$ has already been used. This gives a
maximal matching in the bipartite graph of admissible pairs. If fewer than
$q-2$ pairs were selected, then the unused sets $M'\subseteq M$ and
$U'\subseteq U$ would satisfy $|M'|\ge 3$ and  $|U'|\ge 3$, and no admissible pair would remain between them. Equivalently, the graph
$C[M'\cup U']$ from the proof of \Cref{thm:p2p1-packing} would be complete
bipartite, contradicting the planarity of $C$. Thus the greedy scan finds at
least $q-2$ admissible pairs, and these give pairwise vertex-disjoint induced
copies of $P_2\cup P_1$. The total running time is $O(n^{3/2})+O(n^2)=O(n^2)$.
\end{proof}

\begin{op}\label{ques:algorithmic-p2p1-packing}
Determine the best possible function $f(n)$ such that every plane
triangulation $T$ on $n$ vertices satisfies $\lambda_{P_2\cup P_1}(T)\ge f(n)$.
Moreover, can a packing of size $f(n)$ be constructed in
$O(n)$ time?
\end{op}

\begin{op}\label{ques:complexity-induced-p2p1}
Determine the exact complexity of induced $(P_2\cup P_1)$-packing in plane triangulations.
Given a plane triangulation $T$ and an integer $k$, decide whether
$\lambda_{P_2\cup P_1}(T)\ge k$.
By \Cref{thm:p2p1-packing}, every plane triangulation satisfies
\[
\lambda_{P_2\cup P_1}(T)
\ge
\left\lfloor \frac{|V(T)|}{3}\right\rfloor-2.
\]
Since no packing can have more than $\lfloor |V(T)|/3\rfloor$ copies, the
remaining exact problem is concentrated near the top three possible values.
In particular, can one decide in polynomial time whether
\[
\lambda_{P_2\cup P_1}(T)
=
\left\lfloor \frac{|V(T)|}{3}\right\rfloor
\]
or whether
\[
\lambda_{P_2\cup P_1}(T)
\ge
\left\lfloor \frac{|V(T)|}{3}\right\rfloor-1?
\]
Equivalently, when $3\mid |V(T)|$, what is the complexity of deciding whether
$T$ has an induced $(P_2\cup P_1)$-factor?
\end{op}

\section*{Acknowledgements}

The first, second, and third authors acknowledge the support of the Natural Sciences and Engineering Research Council of Canada (NSERC). The fourth author was supported by JSPS KAKENHI Grant Numbers JP22H00513, JP24H00697, JP25K03076, and JP25K03077.

\bibliographystyle{plainurlnat}
\bibliography{ch.bib}

\end{document}